%% file: main.tex
\documentclass[a4paper,UKenglish,cleveref, autoref, thm-restate]{lipics-v2021}
\usepackage{braket}

\title{An optimal oracle separation of classical and quantum hybrid schemes} 

\titlerunning{An optimal oracle separation of classical and quantum hybrid schemes} 

\author{Atsuya Hasegawa}{Graduate School of Information Science and Technology, The University of Tokyo, Japan}
{atsuyahasegawa@is.s.u-tokyo.ac.jp}{}{}

\author{Fran{\c{c}}ois Le Gall}{Graduate School of Mathematics, Nagoya University, Japan}{legall@math.nagoya-u.ac.jp}{}{}

\authorrunning{A.\,Hasegawa and F.\,Le Gall} 

\Copyright{Atsuya Hasegawa and Fran{\c{c}}ois Le Gall} 



\category{} 

\relatedversion{} 

\funding{JSPS KAKENHI grants Nos.~JP16H01705, JP19H04066, JP20H00579, JP20H04139, JP22J22563 and the MEXT Quantum Leap Flagship Program (MEXT Q-LEAP) grant No.~JPMXS0120319794.}

\acknowledgements{The authors are grateful to Nai-Hui Chia for helpful discussions. They are also grateful to Atul Singh Arora, Alexandru Gheorghiu and Uttam Singh for sharing the information about their result \cite{ags22}.}

\nolinenumbers 

\hideLIPIcs  

\EventEditors{John Q. Open and Joan R. Access}
\EventNoEds{2}
\EventLongTitle{42nd Conference on Very Important Topics (CVIT 2016)}
\EventShortTitle{CVIT 2016}
\EventAcronym{CVIT}
\EventYear{2016}
\EventDate{December 24--27, 2016}
\EventLocation{Little Whinging, United Kingdom}
\EventLogo{}
\SeriesVolume{42}
\ArticleNo{23}

\begin{document}

\maketitle

\begin{abstract}
Recently, Chia, Chung and Lai (STOC 2020) and Coudron and Menda (STOC 2020) have shown that there exists an oracle $\mathcal{O}$ such that $\mathsf{BQP}^\mathcal{O} \neq (\mathsf{BPP^{BQNC}})^\mathcal{O} \cup (\mathsf{BQNC^{BPP}})^\mathcal{O}$. In fact, Chia et al. proved a stronger statement: for any depth parameter $d$, there exists an oracle that separates quantum depth $d$ and $2d+1$, when polynomial-time classical computation is allowed. This implies that relative to an oracle, doubling quantum depth gives classical and quantum hybrid schemes more computational power.

In this paper, we show that for any depth parameter $d$, there exists an oracle that separates quantum depth $d$ and $d+1$, when polynomial-time classical computation is allowed. This gives an optimal oracle separation of classical and quantum hybrid schemes. To prove our result, we consider $d$-Bijective Shuffling Simon's Problem (which is a variant of $d$-Shuffling Simon's Problem considered by Chia et al.) and an oracle inspired by an ``in-place'' permutation oracle.
\end{abstract}

\input{introduction}
\input{preliminary}
\input{d-SSP}
\input{O2H_lemma}
\input{lowerbounds}

\bibliography{lipics-v2021-sample-article}

\end{document}

%% file: introduction.tex
\section{Introduction}
\paragraph*{Background.}
In recent years, the development of quantum computers has been very active (see, e.g., \cite{listqp} for information about current quantum computers) and ``quantum supremacy'' has been claimed \cite{aab+19,zhong2020quantum}. However, it is still difficult to implement large-depth quantum circuits with current quantum technology since such quantum devices are subjective to noise and have short coherent time. One potential way to extract the computational powers of such quantum devices is to consider a hybrid scheme combining them with classical computers. For example, variational quantum algorithms are considered in such a scheme to obtain quantum advantage (see \cite{cab+21} for a survey).

Therefore, understanding the capabilities and limits of this hybrid approach is an essential topic in quantum computation. As one of the most notable results, Cleve and Watrous \cite{cw00} showed the quantum Fourier transformation can be implemented by combining logarithmic-depth quantum circuits with a classical polynomial-time algorithm. With the possibility to implement Shor's algorithm in such a hybrid scheme and the developments of measurement-based quantum computation, Jozsa \cite{joz05} conjectured that “Any quantum polynomial-time algorithm can be implemented with only $O(\log{n})$ quantum depth interspersed with polynomial-time algorithm classical computations”. This can be formalized as $\mathsf{BQP}$ = $\mathsf{BQNC^{BPP}}$. On the other hand, Aaronson \cite{aar05,aar10,aar11} conjectured “there exists an oracle separation between $\mathsf{BQP}$ and $\mathsf{BPP^{BQNC}}$”. $\mathsf{BPP^{BQNC}}$ is a complexity class recognized by a polynomial classical scheme which have access to poly-logarithmic depth quantum circuits. $\mathsf{BQNC^{BPP}}$ and $\mathsf{BPP^{BQNC}}$ are sets of problems recognized by two natural and seemingly incomparable models of hybrid classical and quantum computation.

Recent works by Chia, Chung and Lai \cite{ccl20} and Coudron and Menda \cite{cm20} proved Aaronson’s conjecture and refuted Jozsa’s conjecture in a relativized setting. Interestingly, computational problems and oracles they considered were completely different. Coudron and Menda \cite{cm20} considered, as an oracle problem, the Welded Tree Problem which exhibits a difference between quantum walks and classical random walks: this problem can be solved efficiently by a quantum algorithm \cite{ccd+03} but, in the classical setting, exponential queries are required \cite{ccd+03,fz03}. To prove a lower bound of classical and quantum hybrid schemes, Coudron and Menda introduced ``Information Bottleneck'' to simulate classical and quantum hybrid schemes with fewer classical queries. They showed if we assume the hybrid schemes solve the problem, we reach a contradiction with the lower bound of classical queries from \cite{fz03}.

Chia, Chung and Lai \cite{ccl20} considered $d$-Shuffling Simon’s Problem, which is a variant of Simon’s Problem \cite{sim97}. Since Simon's Problem can be solved with a constant-depth quantum circuit with classical post-processing, we cannot prove the hardness for classical and quantum hybrid schemes. To devise a harder problem, they combine Simon's function with sequential random permutations: for a Simon's function $f$, they consider random one-to-one functions $f_0,...,f_{d-1}$ and two-to-one function $f_d$ such that $f = f_d \circ \cdot \cdot \cdot \circ f_0$. They also hide the domains of the functions in larger domains and apply the idea of the Oneway-to-Hiding (O2H) lemma \cite{ahu18,unr15} to prove the hardness. In fact, they proved a stronger statement below.

\begin{theorem}[\cite{ccl20}]\label{th:ccl20}
  For any $d \in \mathbb{N}$, there exists an oracle $\mathcal{O}$ such that
  \begin{equation}\nonumber
      (\mathsf{BQNC}_{d}^{\mathsf{BPP}})^\mathcal{O} \cup (\mathsf{BPP^{BQNC_{d}}})^\mathcal{O} \neq (\mathsf{BQNC}_{2d+1}^{\mathsf{BPP}})^\mathcal{O} \cap (\mathsf{BPP^{BQNC_{2d+1}}})^\mathcal{O}.
  \end{equation}
\end{theorem}

\paragraph*{Description of our result.}
In this paper, we improve Theorem \ref{th:ccl20} above and show the following result.

\begin{theorem}\label{th:main}
  For any $d \in \mathbb{N}$, there exists an oracle $\mathcal{O}$ such that
  \begin{equation}\nonumber
      (\mathsf{BQNC}_{d}^{\mathsf{BPP}})^\mathcal{O} \cup (\mathsf{BPP^{BQNC_{d}}})^\mathcal{O} \neq (\mathsf{BQNC}_{d+1}^{\mathsf{BPP}})^\mathcal{O} \cap (\mathsf{BPP^{BQNC_{d+1}}})^\mathcal{O}.
  \end{equation}
\end{theorem}

Our result implies that, relative to an oracle, increasing the quantum depth even by \emph{one} gives the hybrid schemes more computational power and it cannot be traded by combining polynomial-time classical processing.

In Theorem \ref{th:main}, quantum circuits consisting of any 1- and 2-qubit gates are considered. Indeed, we give an algorithm by $d+1$-depth quantum circuits consisting only of $\{ H,CNOT \}$ with classical processing for the upper bound (this is also the case for Theorem \ref{th:ccl20} but not mentioned in \cite{ccl20}). Therefore we also prove that, even if we are allowed to use quantum circuits consisting of a restricted gate set contains $\{ H,CNOT \}$ such as Clifford circuits, adding even one quantum depth gives the two hybrid schemes more computational power relative to an oracle.

\newpage
\paragraph*{Outline of our approach.}
Chia et al. gave the upper bound $(\mathsf{BQNC}_{2d+1}^{\mathsf{BPP}})^\mathcal{O} \cap (\mathsf{BPP^{BQNC_{2d+1}}})^\mathcal{O}$ for $d$-Shuffling Simon's Problem by an algorithm inspired by the Simon's algorithm. Since they considered a standard oracle, $U_f \ket{x} \ket{0} = \ket{x} \ket{f(x)}$, it is required to erase the information of past queries and it takes $d$-quantum depth. To eliminate the $d$-quantum depth, we propose an idea to consider an ``in-place'' permutation oracle \cite{aar02,fk18} acts as $U_f \ket{x} = \ket{f(x)}$. However, when $f$ is a Simon's function, $f_d$ on a restricted domain is also a two-to-one function and there is no unitary operator $U_{f_d}$ such that $U_{f_d}\ket{x}=\ket{f_d(x)}$. Therefore, in this paper, we consider another function $\eta$ and make the function bijective. We name the problem $d$-Bijective Shuffling Simon's Problem and show an upper bound $(\mathsf{BQNC}_{d+1}^{\mathsf{BPP}})^\mathcal{O} \cap (\mathsf{BPP^{BQNC_{d+1}}})^\mathcal{O}$. The other obstacle is, for $f_d$ and the shadows to prove the lower bounds, how to define a unitary operator that includes mappings to $\perp$ (a constant with no information). Note that this is because there exists no unitary operator $U_{\perp}$ such that $U_{\perp}\ket{x}=\ket{\perp}$. In this paper, we give a solution by keeping values on domains and considering ``flags'' on ancilla qubits. Finally we carefully tailor the Oneway-to-Hiding lemma in our quantum oracle setting and show that the similar proofs of the lower bounds also follow as \cite{ccl20}.

\paragraph*{Related work.}
Arora, Gheorghiu and Singh \cite{ags22} proved oracle separations of $(\mathsf{BQNC}_{d}^{\mathsf{BPP}})^\mathcal{O}$ and $(\mathsf{BPP^{BQNC_{d}}})^\mathcal{O}$ with respect to each other. As corollaries, they obtained sharper separations than \cite{ccl20} for each scheme. For the quantum-classical scheme, they proved an oracle separation between quantum depth $d$ and $d+1$ if the Hadamard measurements are allowed in every layer. In our result, we only need to measure qubits in the Hadamard basis in the last layer. For the classical-quantum scheme, they proved a separation between quantum depth $d$ and $d+5$ relative to what they call a stochastic oracle (which is non-unitary). Our separation is between quantum depth $d$ and $d+1$ relative to a unitary oracle.

In an independent work \cite{ch22}, Chia and Hung have also shown how to reduce the gap from $d$ versus $2d+1$ to $d$ versus $d+1$ by techniques similar to ours (they consider an oracle inspired by an ``in-place'' permutation oracle and manage to make the final function one-to-one). They also instantiate the oracle separation to construct a protocol such that a classical verifier can check if a prover has a quantum depth of at least $d+1$. 

\paragraph*{Organization of the paper.}
After giving preliminaries in Section 2, in Section 3 we define the oracle problem that we call $d$-Bijective Shuffling Simon's Problem and give the upper bound $(\mathsf{BQNC}_{d+1}^{\mathsf{BPP}})^\mathcal{O} \cap (\mathsf{BPP^{BQNC_{d+1}}})^\mathcal{O}$. In Section 4, we prove the Oneway-to-Hiding lemma for our quantum oracle and, in Section~5, we prove the lower bounds for $(\mathsf{BQNC}_d^{\mathsf{BPP}})^\mathcal{O}$ and  $(\mathsf{BPP}^{\mathsf{BQNC}_d})^\mathcal{O}$.

The main contribution of our work is to define the $d$-Bijective Shuffling Simon's Problem and give the upper bound with quantum depth $d+1$ (Section \ref{3}). The proof of the lower bound (Section \ref{5}) is very similar to \cite{ccl20} except the Oneway-to-Hiding lemma (Section \ref{4}), which has to be adapted to the quantum oracle of this paper. We are grateful to Nai-Hui Chia for discussions about this, and in particular for clarifying that all steps in the lower bound from~\cite{ccl20} remain true for our new oracle as well, with the exception of this Oneway-to-Hiding lemma.

%% file: preliminary.tex
\newpage
\section{Preliminaries}
\subsection{State distances}
Let us recall some notions about the distances of quantum states \cite{nc10}.
\begin{definition}
    For any two mixed states $\rho$ and $\sigma$,
    \begin{itemize}
        \item (Fidelity) $F(\rho,\sigma) :=  \mathrm{tr}(\sqrt{\sqrt{\rho}\sigma\sqrt{\rho}})$.
        \item (Bures distance) $B(\rho,\sigma) := \sqrt{2-2F(\rho,\sigma)}$.
    \end{itemize}
\end{definition}

\begin{claim}
    For any two mixed states $\rho$ and $\sigma$, any quantum algorithm $\mathcal{A}$ and any classical string s, 
    \begin{equation}\nonumber
        |\mathrm{Pr}[\mathcal{A}(\rho)=s]-\mathrm{Pr}[\mathcal{A}(\sigma)=s]| \leq B(\rho,\sigma).
    \end{equation}
\end{claim}

\subsection{Computational models}
Let us introduce computational models in this paper, especially two relativized hybrid schemes of quantum and classical computations.
\subsubsection{Quantum circuits}
First, we define $d$-depth quantum circuits and promise problems solved by the computational schemes. We refer to \cite{nc10} for a standard reference of quantum computation and circuits. We define a one-layer (depth) unitary is a set of one- and two-qubit gates act on disjoint sets of qubits. A promise problem denotes a pair $L = (L_\mathrm{yes},L_\mathrm{no})$, where $L_\mathrm{yes}, L_\mathrm{no} \subseteq \Sigma^*$ are sets of strings satisfying $L_\mathrm{yes} \cap L_\mathrm{no} = \emptyset$.

\begin{definition}[$d$-depth quantum circuit, $\mathsf{QNC}_d$]
    $\mathsf{QNC}_d$ is the class of all quantum circuit family $\{C_n : n \in \mathbb{N} \}$ satisfying that $C_n$ acts on n input qubits and $poly(n)$ ancilla qubits, and consists of successive $d$-layers of arbitrary one- and two-qubit gates, $U_1 U_2...U_d$, and measure all qubits in the computational basis.
\end{definition}

\begin{definition}[$\mathsf{BQNC}_d$]
    A promise problem $L$ is in $\mathsf{BQNC}_d$ if and only if there exists a circuit family $\{C_n : n \in \mathbb{N} \} \in \mathsf{QNC}_d$ satisfying the following properties:
    \begin{itemize}
        \item for all $x \in L_{\mathrm{yes}}$, $\mathrm{Pr}[C_{|x|}(x)=1] \geq \frac{2}{3}$;
        \item for all $x \in L_\mathrm{no}$, $\mathrm{Pr}[C_{|x|}(x)=1] \leq \frac{1}{3}$.
    \end{itemize}
\end{definition}

In the definition above, we consider arbitrary one- and two-qubit gates. In this paper, we also consider a gate set $\{H,CNOT\}$, where $H=\frac{1}{\sqrt{2}}\left(\begin{smallmatrix}1&1\\1&-1\end{smallmatrix}\right)$ is the Hadamard gate and $CNOT=\left(\begin{smallmatrix}1&0&0&0\\0&1&0&0\\0&0&0&1\\0&0&1&0\end{smallmatrix}\right)$ is the controlled-not gate. A quantum circuit consisting only of Clifford gates $\{H,S,CNOT\}$, where $S=\left(\begin{smallmatrix}1&0\\0&i\end{smallmatrix}\right)$ is the phase gate, is called a Clifford circuit.

For an oracle $\mathcal{O}$, let us define $(\mathsf{QNC}_d)^\mathcal{O}$ be a computational scheme similar to $\mathsf{QNC}_d$ considering $U_{d+1}\mathcal{O}U_d...\mathcal{O}U_1$ as layers and $(\mathsf{BQNC}_d)^\mathcal{O}$ be a set of promise problems solved by $(\mathsf{QNC}_d)^\mathcal{O}$ with high probability. Note that for relativized $d$-depth quantum circuit, we consider to add an extra single layer to process the final oracle access following \cite{ccl20}. We refer to Figure \ref{qnc} for an illustration.

\begin{figure}[htbp]
  \centering
  \includegraphics[clip,width=8.0cm]{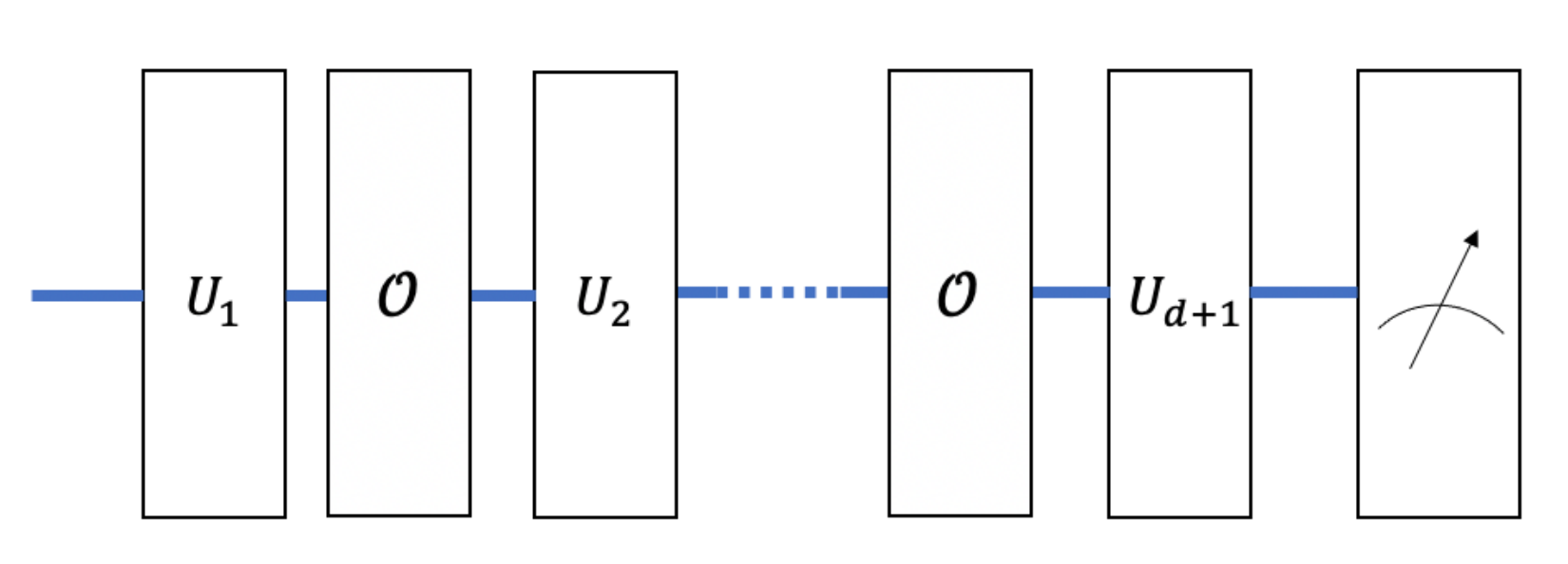}
  \caption{$d$-depth quantum circuit which can access to an oracle $\mathcal{O}$.}
  \label{qnc}
\end{figure}

\subsubsection{Quantum-classical hybrid schemes}
Next, we consider the $d$-quantum-classical scheme, which is a generalized model for $d$-depth measurement-based quantum computation. This is the same model as in \cite{ccl20}. The scheme is denoted as $d$-QC scheme and we represent it as the following sequence:
\begin{equation}\nonumber
    \left( \mathcal{A}_c \xrightarrow{c} (\textstyle{\prod_{0/1}} \otimes I) \circ U_1 \right) \xrightarrow{c,q}  \left( \mathcal{A}_c \xrightarrow{c} (\textstyle{\prod_{0/1}} \otimes I) \circ U_2  \right) \xrightarrow{c,q} \left( \mathcal{A}_c \cdot \cdot \cdot \circ U_d \right) \xrightarrow{c} \mathcal{A}_c,
\end{equation}
where $\mathcal{A}_c$ is a classical probabilistic polynomial-time algorithm, $U_i$ is a one-depth quantum layer and $\textstyle{\prod_{0/1}}$ is a measurement in the computational basis or the Hadamard basis\footnote{Without the Hadamard measurements, our upper bound becomes quantum depth $d+2$ for the quantum-classical hybrid scheme. The upper bound of Theorem \ref{th:ccl20} for the scheme also becomes quantum depth $2d+2$ without the Hadamard measurements.}. The arrows $\xrightarrow{c}$ and $\xrightarrow{q}$ represent transmissions of polynomial-size classical and quantum bits. Between quantum layers $U_i$ and $U_{i+1}$, $\mathcal{A}_c$ can use measurement results of an arbitrary part of qubits after $U_i$ applies and send classical polynomial-size information to $U_{i+1}$. We refer to Figure \ref{d-qc} for an illustration. 

\begin{figure}[htbp]
  \centering
  \includegraphics[clip,width=8.0cm]{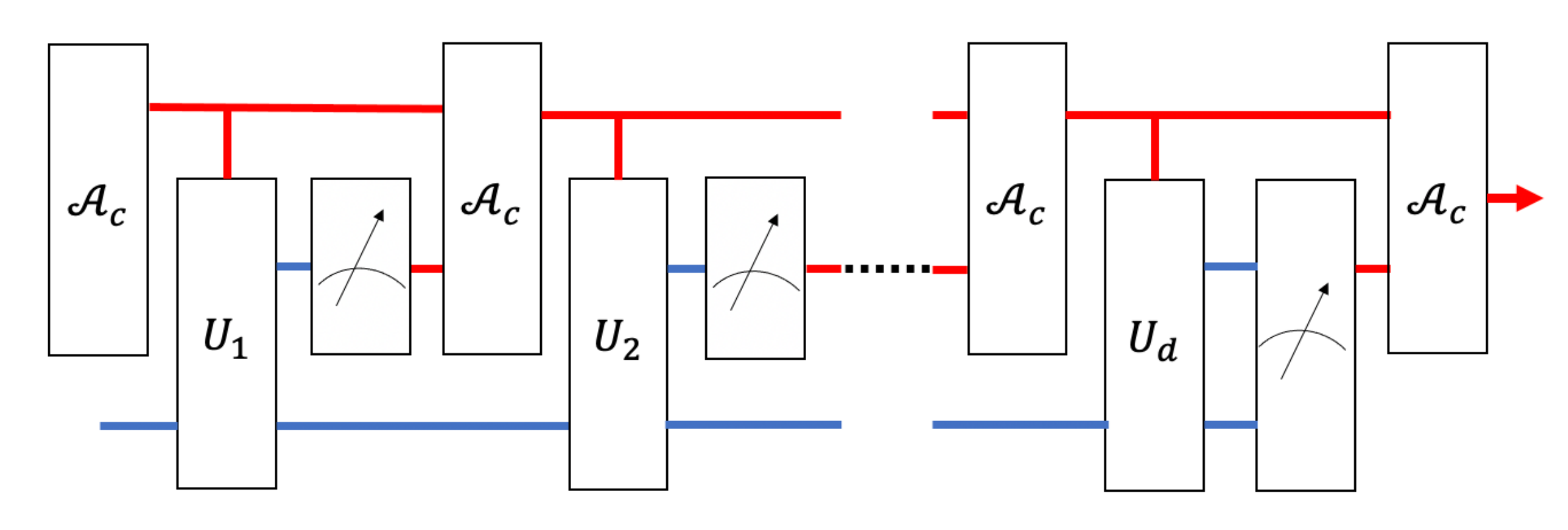}
  \caption{The $d$-QC scheme. Red lines stand for classical wires and blue ones stand for quantum wires.}
  \label{d-qc}
\end{figure}

Let us consider a relativized version of the scheme. Let $\mathcal{A}^\mathcal{O}$ be a $d$-QC scheme with access to an oracle $\mathcal{O}$. We represent the relativized $d$-QC scheme $\mathcal{A}^\mathcal{O}$ as a sequence of operators:
\begin{equation}\nonumber
    (L_1)^\mathcal{O} \xrightarrow{c,q} \cdot \cdot \cdot \xrightarrow{c,q} (L_d)^\mathcal{O} \xrightarrow{c} \mathcal{A}_c^\mathcal{O},
\end{equation}
where $\mathcal{A}_c^\mathcal{O}$ is a classical polynomial-time algorithm which can query to the oracle $\mathcal{O}$, and $(L_i)^\mathcal{O} := \mathcal{A}_c^\mathcal{O} \xrightarrow{c} (\textstyle{\prod_{0/1}} \otimes I) \circ \mathcal{O} U_i$. We refer to Figure \ref{d-qc^o} for an illustration. Then, we define the set of promise problems which can be solved by the relativized $d$-QC schemes.

\begin{figure}[htbp]
  \centering
  \includegraphics[clip,width=10.0cm]{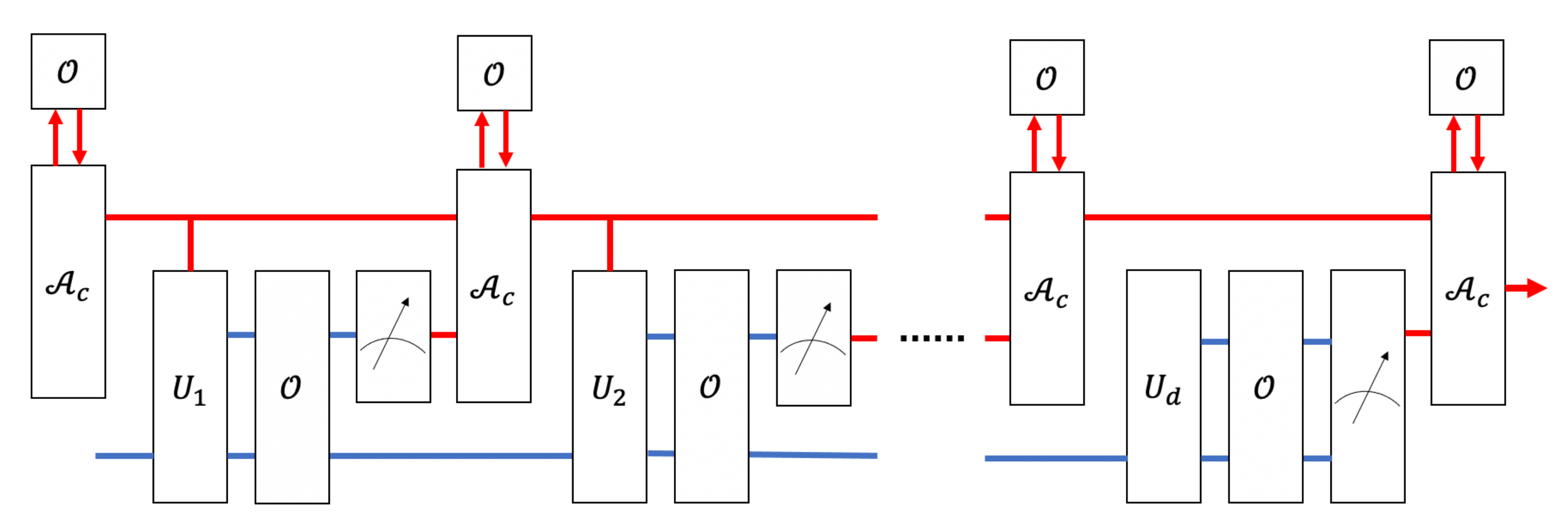}
  \caption{The $d$-QC scheme with access to an oracle $\mathcal{O}$.}
  \label{d-qc^o}
\end{figure}

\begin{definition}[$(\mathsf{BQNC}_d^{\mathsf{BPP}})^\mathcal{O}$, Definition 3.8 in \cite{ccl20}]
    A promise problem $L$ is in $(\mathsf{BQNC}_d^{\mathsf{BPP}})^\mathcal{O}$ if and only if there exists a family of relativized $d$-QC schemes $\{\mathcal{A}^\mathcal{O}_n:n \in \mathbb{N} \}$ satisfying the following properties:
    \begin{itemize}
        \item for all $x \in L_{\mathrm{yes}}$, $\mathrm{Pr}[\mathcal{A}_{|x|}^\mathcal{O}(x)=1] \geq \frac{2}{3}$;
        \item for all $x \in L_{\mathrm{no}}$, $\mathrm{Pr}[\mathcal{A}_{|x|}^\mathcal{O}(x)=1] \leq \frac{1}{3}$. 
    \end{itemize}
\end{definition}

\subsubsection{Classical-quantum hybrid schemes}
Finally, we define the $d$-classical-quantum scheme, which is a classical polynomial-time algorithm which has access to a $d$-depth quantum circuit during the computation (up to polynomial times). This is the same model as in \cite{ccl20}. The scheme is denoted as $d$-CQ scheme and represented as follows:
\begin{equation}\nonumber
    \mathcal{A}_{c,1} \xrightarrow{c} \textstyle{\prod_{0/1}} \circ U_{d} \cdot \cdot \cdot U_{1} \xrightarrow{c} \cdot \cdot \cdot \xrightarrow{c} \mathcal{A}_{c,m-1} \xrightarrow{c} \textstyle{\prod_{0/1}} \circ U_{d} \cdot \cdot \cdot U_{1} \xrightarrow{c} \mathcal{A}_{c,m},
\end{equation}
where $m$ is a polynomial in $n$, $\mathcal{A}_{c,i}$ is an $i$th classical probabilistic polynomial-time algorithm, $U_i$ is a one-depth quantum layer, and $\textstyle{\prod_{0/1}}$ is the computational basis measurement. Each $\mathcal{A}_{c,i+1}$ can depend on the polynomial-size information sent from $\mathcal{A}_{c,i}$ and measurement results of the $d$-depth quantum circuit $\mathcal{A}_{c,i}$ calls. We refer to Figure \ref{d-cq} for an illustration.

\begin{figure}[htbp]
  \centering
  \includegraphics[clip,width=12.0cm]{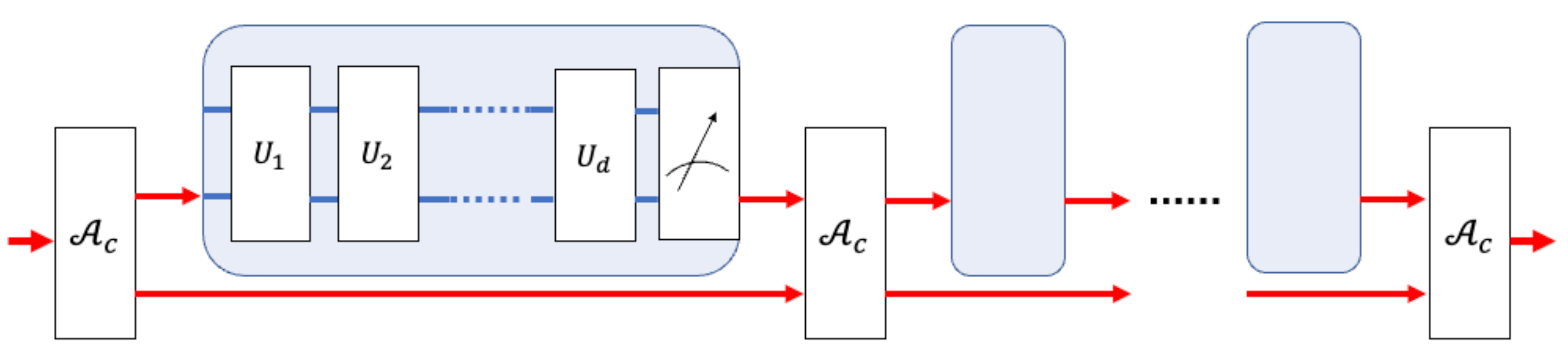}
  \caption{The $d$-CQ scheme. Red lines stand for classical wires and blue ones stand for quantum wires.}
  \label{d-cq}
\end{figure}

Let us consider a relativized version of the scheme. Let $\mathcal{A}^\mathcal{O}$ be a $d$-CQ scheme with an access to an oracle $\mathcal{O}$. We represent $\mathcal{A}^\mathcal{O}$ as follows:
\begin{equation}\nonumber
    (L_1)^\mathcal{O} \xrightarrow{c} (L_2)^\mathcal{O} \xrightarrow{c} \cdot \cdot \cdot \xrightarrow{c} (L_{m-1})^\mathcal{O} \xrightarrow{c} (\mathcal{A}_{c,m})^\mathcal{O},
\end{equation}
where $\mathcal{A}_{c,i}$ is an $i$th classical polynomial-time algorithm which can query to the oracle $\mathcal{O}$, and $(L_i)^\mathcal{O} := (\mathcal{A}_{c,i})^\mathcal{O} \xrightarrow{c} \textstyle{\prod_{0/1}} \circ (U_{d+1}\mathcal{O}U_{d} \cdot \cdot \cdot \mathcal{O}U_{1})$. We refer to Figure \ref{d-cq^o} for an illustration. Then, we define the set of promise problems which can be solved by the relativized $d$-CQ schemes.

\begin{definition}[$(\mathsf{BPP}^{\mathsf{BQNC}_d})^\mathcal{O}$, Definition 3.10 in \cite{ccl20}]
    A promise problem L is in $(\mathsf{BPP}^{\mathsf{BQNC}_d})^\mathcal{O}$ if and only if there exists a family of relativized $d$-CQ schemes $\{\mathcal{A}^\mathcal{O}_n:n \in \mathbb{N} \}$ satisfying the following properties:
    \begin{itemize}
        \item for all $x \in L_{\mathrm{yes}}$, $\mathrm{Pr}[\mathcal{A}_{|x|}^\mathcal{O}(x)=1] \geq \frac{2}{3}$;
        \item for all $x \in L_{\mathrm{no}}$, $\mathrm{Pr}[\mathcal{A}_{|x|}^\mathcal{O}(x)=1] \leq \frac{1}{3}$. 
    \end{itemize}
\end{definition}

\begin{figure}[htbp]
  \centering
  \includegraphics[clip,width=14.0cm]{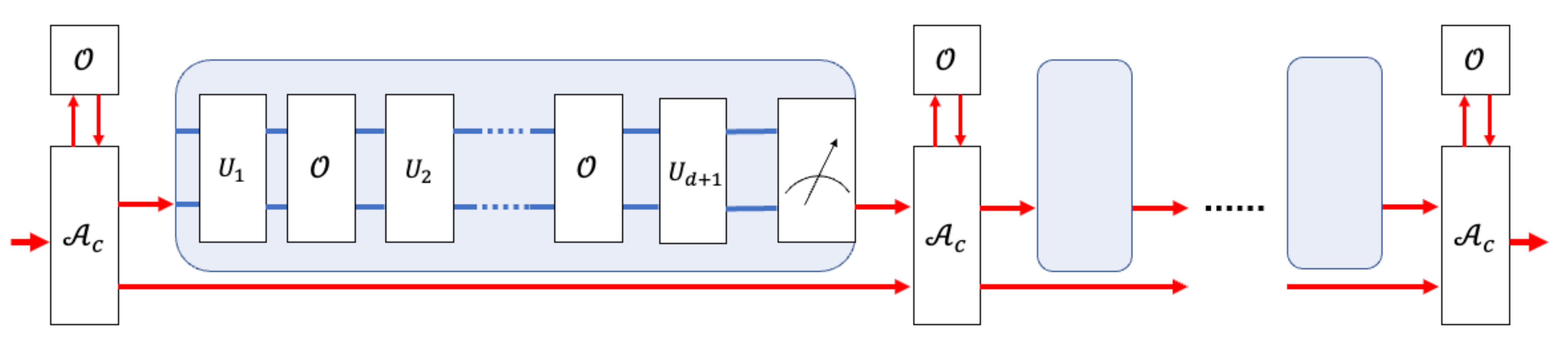}
  \caption{The $d$-CQ scheme with access to an oracle $\mathcal{O}$.}
  \label{d-cq^o}
\end{figure}

\subsection{Simon's problem}
Let us recall the definitions of Simon's function and Simon's problem \cite{sim97}.
\begin{definition}[Simon's function]
    A two-to-one function $f:\mathbb{Z}_2^n \rightarrow \mathbb{Z}_2^n$ for $n \in \mathbb{N}$ is a Simon's function if there exists a period $s \in \mathbb{Z}_2^n$ such that $f(x)=f(x \oplus s)$ for $x \in \mathbb{Z}_2^n$. Let $\mathbf{F}$ be the set of all Simon's functions from $\mathbb{Z}_2^n$ to $\mathbb{Z}_2^n$.
\end{definition}

\begin{definition}[Simon's problem]
    Given $f$ chosen from $\mathbf{F}$ uniformly at random, the problem is to obtain the period $s$.
\end{definition}

\begin{definition}[Decision Simon's problem]
     Given $f$ to be either a random Simon's function from $\mathbf{F}$ or a random one-to-one function from $\mathbb{Z}_2^n$ to $\mathbb{Z}_2^n$ with equal probability, the promise problem is to distinguish the two cases.
\end{definition}

%% file: d-SSP.tex
\section{$d$-Bijective Shuffling Simon's Problem ($d$-BSSP)}\label{3}
In this section, we introduce the oracle problem, the $d$-Bijective Shuffling Simon's Problem, which is abbreviated as $d$-$\mathsf{BSSP}$ in the rest of the paper. It is similar to $d$-Shuffling Simon's Problem ($d$-$\mathsf{SSP}$) in \cite{ccl20} but there are several modifications of quantum oracles to make the upper bound $(\mathsf{BPP^{BQNC_{d+1}}})^\mathcal{O} \cap (\mathsf{BQNC_{d+1}^{BPP}})^\mathcal{O}$. 

Let us consider shufflings of functions. For any two set $X$ and $Y$, let $P(X,Y)$ be the set of one-to-one functions from $X$ to $Y$. In this work, it is enough to consider random shufflings of one-to-one functions and Simon's functions. 

\begin{definition}[$(d,f)$-shuffling, Definition 4.1 in \cite{ccl20}]
    Let $d \in \mathbb{N}$ and $f$ be a one-to-one function or a Simon's function from $\mathbb{Z}_2^n$ to $\mathbb{Z}_2^n$. A $(d,f)$-shuffling is $\mathcal{F}:=(f_0,f_1,...,f_d)$, where $f_0,...,f_{d-1}$ are chosen uniformly at random from $P(\mathbb{Z}_2^{(d+2)n},\mathbb{Z}_2^{(d+2)n})$ and  $f_d:\mathbb{Z}_2^{(d+2)n} \rightarrow \mathbb{Z}_2^{(d+2)n}$ is a function satisfying the following properties: Let $S_d :=\{f_{d-1} \circ \cdot \cdot \cdot \circ f_0(x') : x' = 0,...,2^n-1\}$.
    \begin{itemize}
        \item For $x \in S_d$, $f_d(x)=f \circ f_0^{-1} \circ \cdot \cdot \cdot \circ f_{d-1}^{-1}(x)$.
        \item For $x \notin S_d$, $f_d(x)=\perp$.
    \end{itemize}
    Let $\mathbf{SHUF}(d,f)$ be a set of all $(d,f)$-shuffling functions.
\end{definition}
Note that the way to choose $f_d$ is unique for each $f_0,...,f_{d-1}$.

We consider a standard oracle for $f_0$ and ``in-place'' permutation oracles for $f_1,...,f_{d-1},f_d$. When $f$ is a Simon's function, $f_d$ on $S_d$ is also two-to-one and there exists no quantum oracle to implement $f_d$ ``in-place'' since a two-to-one function is not a unitary operator. Therefore, we consider a function $\eta$ to make the mapping bijective. We also consider a function $\zeta$ which represents whether $x_d$ is in $S_d$. 

\begin{definition}[$(d,f)$-bijective shuffling]
     For a $(d,f)$-shuffling $\mathcal{F} =(f_0,f_1,...,f_d)$, a $(d,f)$-bijective shuffling is  $\mathcal{F}_b:=(f_0,f_1,...,f_{d-1},f_d',\zeta,\eta)$, where $f_d':\mathbb{Z}_2^{(d+2)n} \rightarrow \mathbb{Z}_2^{(d+2)n}$, $\zeta:\mathbb{Z}_2^{(d+2)n} \rightarrow \mathbb{Z}_2$ and $\eta:\mathbb{Z}_2^{(d+2)n} \rightarrow \mathbb{Z}_2$ satisfying the following properties: Let $S_d :=\{f_{d-1} \circ \cdot \cdot \cdot \circ f_0(x') : x' = 0,...,2^n-1\}$.
    \begin{itemize}
        \item For $x \in S_d$, $f_d'(x)=f_d(x)$ and $\zeta(x)=1$
        \item For $x \notin S_d$, $f_d'(x)=x$ and $\zeta(x)=0$
        \item  If $f$ is a Simon's function,
        \begin{itemize}
            \item For $x \in S_d$, $\eta(x) = \eta(x') \oplus 1$ where $x'$ is the unique element in $S_d$ such that $f_d(x)=f_d(x')$.
            \item For $x \notin S_d$, $\eta(x)=1$
        \end{itemize}
        Otherwise,
        \begin{itemize}
            \item For $x \in S_d$, $\eta(x)$ is chosen uniformly at random from $[0,1]$
            \item For $x \notin S_d$, $\eta(x)=1$
        \end{itemize}
    \end{itemize}
    Let $\mathbf{BSHUF}(d,f)$ be a set of all $(d,f)$-bijective shuffling functions.
\end{definition}

For fixed $d \in \mathbb{N}$ and a function $f:\mathbb{Z}_2^n \rightarrow \mathbb{Z}_2^n$, we can define a random oracle which is a $(d,f)$-bijective shuffling chosen uniformly randomly from $\mathbf{BSHUF}(d,f)$.

\begin{definition}[Bijective shuffling oracle $\mathcal{O}_{\mathrm{unif}}^{f,d}$]
    Let $d,n \in \mathbb{N}$. Let $f$ be a Simon's function or a one-to-one function from $\mathbb{Z}_2^n$ to $\mathbb{Z}_2^n$. The bijective shuffling oracle $\mathcal{O}_{\mathrm{unif}}^{f,d}$ is a $(d,f)$-bijective shuffling $\mathcal{F}_b$ uniformly chosen from $\mathbf{BSHUF}(d,f)$.
\end{definition}

If we sample a $(d,f)$-bijective shuffling uniformly randomly from $\mathbf{BSHUF}(d,f)$, the permutations $f_0,...,f_{d-1}$ are uniformly distributed in $P(\mathbb{Z}_2^{(d+2)n},\mathbb{Z}_2^{(d+2)n})$ independently of $f$. This is because for each $(d,f)$-shuffling, $f_d'$ and $\zeta$ are uniquely chosen and the number of ways to choose $\eta$ and $(d,f)$-bijective shufflings is the same. From the definition of $\mathbf{BSHUF}(d,f)$, for such a $(d,f)$-bijective shuffling $\mathcal{F}_b$, only $f_d'$ and $\eta$ on $S_d$ encodes the information of $f$.

Next, let us define the quantum oracle access to the bijective shuffling oracle $\mathcal{O}_{\mathrm{unif}}^{f,d}$. In this paper, as in \cite{ccl20}, we represent the input quantum state $\ket{\phi}$ to $\mathcal{O}_{\mathrm{unif}}^{f,d}$ as follows:
\begin{equation}\nonumber
    \ket{\phi} := \sum_{\mathbf{X}_0,...,\mathbf{X}_d} c(\mathbf{X}_0,...,\mathbf{X}_d) \left(\bigotimes_{i=0}^d \ket{i,\mathbf{X}_i}\right)_{\mathbf{R}_Q} \otimes \ket{0}_{\mathbf{R}_N} \otimes \ket{w(\mathbf{X}_0,...,\mathbf{X}_d)}_{\mathbf{R}_W},
\end{equation}
where $\mathbf{X}_i$ is a set of elements in the domain of the functions $f_0,...,f_{d-1},f'_d$ and $c(\mathbf{X}_0,...,\mathbf{X}_d)$ is an arbitary coefficient and $\ket{w(\mathbf{X}_0,...,\mathbf{X}_d)}$ is an arbitrary working state for quantum layers, $U_i$s. By the access to the quantum oracle, the parallel queries in the register $\mathbf{R}_Q$ and the ancilla qubits in the register $\mathbf{R}_N$ are processed, while the remaining qubits in the register $\mathbf{R}_W$ are unchanged.

For $\mathcal{F}_b \in \mathbf{BSHUF}(d,f)$, we define
\begin{eqnarray*}\nonumber
    \mathcal{F}_b\ket{\phi} := \sum_{\mathbf{X}_0,...,\mathbf{X}_d} \mspace{-10mu} c(\mathbf{X}_0,...,\mathbf{X}_d) \mspace{-35mu} & & \biggl(\ket{0,\mathbf{X}_0}\ket{f_0(\mathbf{X}_0)} \otimes \bigotimes_{i=1}^{d-1} \ket{i, f_i(\mathbf{X}_i)}\ket{0} \\  \mspace{-80mu} & & \otimes \ket{d, f'_d(\mathbf{X}_d)}\ket{\zeta(\mathbf{X}_d)}\ket{\eta(\mathbf{X}_d)}\ket{0} \biggl)_{\mathbf{R}_Q,\mathbf{R}_N} \otimes \ket{w(\mathbf{X}_0,...,\mathbf{X}_d)}_{\mathbf{R}_W}.
\end{eqnarray*}
Note that the ancilla qubits $\ket{0}$ are required to define the shadow in Section 4. We also define applying $\mathcal{O}_{\mathrm{unif}}^{f,d}$ on $\ket{\phi}$ as
\begin{equation}\nonumber
    \mathcal{O}_{\mathrm{unif}}^{f,d} (\ket{\phi}\bra{\phi}) := \sum_{\mathcal{F}_b \in \mathbf{BSHUF}(d,f)} \frac{1}{|\mathbf{BSHUF}(d,f)|} \mathcal{F}_b \ket{\phi}\bra{\phi} \mathcal{F}_b.
\end{equation}

Now we can define the $d$-Bijective Shuffling Simon's Problem ($d$-$\mathsf{BSSP}$).

\begin{definition}[The search $d$-Bijective Shuffling Simon's Problem]
     Let $d \in \mathbb{N}$ and $f$ be a random Simon's function. Given the $(d,f)$-bijective shuffling oracle $\mathcal{O}_{\mathrm{unif}}^{f,d}$, the problem is to obtain the period $s$.
\end{definition}

\begin{definition}[The decision $d$-Bijective Shuffling Simon's Problem]
    Let $d \in \mathbb{N}$ and $f$ be either a random Simon's function or a random one-to-one function with equal probability. Given the $(d,f)$-bijective shuffling oracle $\mathcal{O}_{\mathrm{unif}}^{f,d}$, the promise problem is to distinguish the two cases.
\end{definition}

We show that the $d$-$\mathsf{BSSP}$ can be solved with $d+1$-depth quantum circuits and classical processing.
\begin{theorem}
 The search $d$-$\mathsf{BSSP}$ can be solved with a $(d+1)$-depth quantum circuit composed of $\{H, CNOT\}$ with polynomial-time classical processing. Therefore it can be solved with the $(d+1)$-QC scheme and the $(d+1)$-CQ scheme, and the decision $d$-$\mathsf{BSSP}$ is in $(\mathsf{BPP^{BQNC_{d+1}}})^\mathcal{O} \cap (\mathsf{BQNC}_{d+1}^{\mathsf{BPP}})^\mathcal{O}$.
\end{theorem}
\begin{proof}
The $d$-$\mathsf{BSSP}$ can be solved by the following algorithm which is inspired by the Simon's algorithm and queries to the quantum oracle $d+1$ times.
\begin{eqnarray*}
 &&\mspace{-80mu} \ket{0} \ket{0} \ket{0} \ket{0} \hspace{5mm} \xrightarrow{H^{\otimes n}} \hspace{5mm} \sum_{x \in \mathbb{Z}_2^n} \ket{x} \ket{0} \ket{0} \ket{0} \hspace{5mm} \xrightarrow{f_0} \hspace{5mm} \sum_{x \in \mathbb{Z}_2^n} \ket{x} \ket{f_0(x)} \ket{0} \ket{0}\\
 &\xrightarrow{f_1}& \sum_{x \in \mathbb{Z}_2^n} \ket{x} \ket{f_1(f_0(x))} \ket{0} \ket{0} \hspace{5mm} \xrightarrow{f_2} \hspace{5mm} \cdot \cdot \cdot \\
 &\xrightarrow{f_{d-1}}& \sum_{x \in \mathbb{Z}_2^n} \ket{x} \ket{f_{d-1}( \cdot \cdot \cdot f_1(f_0(x)))} \ket{0} \ket{0}\\
 &\xrightarrow{f_d',\zeta,\eta}& \sum_{x \in \mathbb{Z}_2^n} \ket{x} \ket{f(x)} \ket{\zeta(f_{d-1}( \cdot \cdot \cdot f_1(f_0(x))))} \ket{\eta(f_{d-1}( \cdot \cdot \cdot f_1(f_0(x))))} \\
 &\xrightarrow{Measure}& \mspace{-10mu} \left[ \ket{x} \ket{\eta(f_{d-1}( \cdot \cdot \cdot f_1(f_0(x))))} + \ket{x \oplus s} \ket{\eta(f_{d-1}( \cdot \cdot \cdot f_1(f_0(x \oplus s))))} \right] \ket{f(x)} \ket{1}\\
 &\xrightarrow{H^{\otimes n+1}}& \mspace{-30mu} \sum_{j \in \mathbb{Z}_2^n} \left[ (-1)^{x \cdot j} (1 + (-1)^{s \cdot j}) \ket{j} \ket{0} + (-1)^{(x \cdot j) \oplus \eta(f_{d-1}( \cdot \cdot \cdot f_1(f_0(x))))} (1 - (-1)^{s \cdot j}) \ket{j} \ket{1} \right]
\end{eqnarray*}
When the measurement result of the last qubit is 0, then $s \cdot j = 0$. Otherwise, $s \cdot j = 1$. Therefore, by sampling $O(n)$ times and solving linear equations, we can obtain the period $s$ with high probability.
\end{proof}

%% file: O2H_lemma.tex
\section{Analyzing the Bijective Shuffling Oracle and Oneway-to-Hiding (O2H) Lemma}\label{4}
In this section, we define several notations of the bijective shuffling problem and prove Oneway-to-Hiding (O2H) Lemma \cite{ahu18} for our quantum oracle defined in the previous section. First, let us introduce notations of hidden sets.

\begin{definition}\label{over_s}
    Let $S_0=\{0,...,2^n-1\}$. For $j=1,...,d$, let $S_{j} = f_{j-1} \circ \cdot \cdot \cdot f_0(S_0)$.
\end{definition}

\newpage
\begin{definition}[The hidden sets $\mathbf{S}$, Definition 5.2 in \cite{ccl20}]\label{hiddenset}
    Let $\mathcal{F}_b$ be a $(d,f)$-bijective shuffling. The sequence of hidden sets $\mathbf{S} = (\overline{S}^{(0)},...,\overline{S}^{(d)})$ is defined as follows:
    \begin{itemize}
        \item Let $S_j^{(0)} = \mathbb{Z}_2^{(d+2)n}$ for $j=0,...,d$. $\overline{S}^{(0)} := (S_0^{(0)},...,S_0^{(d)})$.
        \item For $l = 1,...,d$, for $j=l,...,d$, we choose $S_j^{(l)} \subseteq S_j^{(l-1)}$ randomly satisfying that $\frac{|S_j^{(l)}|}{|S_j^{(l-1)}|} \leq \frac{1}{2^n}$, $f_{j-1}(S_{j-1}^{(l)})=S_j^{(l)}$, and $S_j \subseteq S_j^{(l)}$. $\overline{S}^{(l)} := (S_l^{(l)},...,S_d^{(l)})$.
    \end{itemize}
\end{definition}

\begin{definition}[$\mathcal{F}^{(l)}$ and $\hat{\mathcal{F}}^{(l)}$]\label{f_l}
    For $l=1,...,d-1$, let $f_j^{(l)}$ be $f_j$ on $S_j^{(l-1)} \setminus S_j^{(l)}$ and $\hat{f}_j^{(l)}$ be $f_j$ on $S_j^{(l)}$. Let $f_d'^{(l)}$ be $f_d'$ on $S_d^{(l-1)} \setminus S_d^{(l)}$ and $\hat{f}_d'^{(l)}$ be $f_d'$ on $S_d^{(l)}$. For $l=1,...,d$, let $\zeta^{(l)}$ be $\zeta$ on $S_d^{(l-1)} \setminus S_d^{(l)}$ and $\hat{\zeta^{(l)}}$ be $\zeta$ on $S_d^{(l)}$. Also for $l=1,...,d$, let $\eta^{(l)}$ be $\eta$ on $S_d^{(l-1)} \setminus S_d^{(l)}$ and $\hat{\eta}^{(l)}$ be $\eta$ on $S_d^{(l)}$. Then, we define $\mathcal{F}^{(1)} := (f_0,f_1^{(1)},...,f_d'^{(1)},\zeta^{(1)},\eta^{(1)})$, $\mathcal{F}^{(d+1)} := (\hat{f}_d'^{(d)},\hat{\zeta}^{(d)},\hat{\eta}^{(d)})$, and $\mathcal{F}^{(l)} := (\hat{f}_{l-1}^{(l)},f_l^{(l)},...,f_d'^{(l)},\zeta^{(l)},\eta^{(l)}))$ for $l=2,...,d$. Also, we define $\hat{\mathcal{F}}^{(0)} := (f_0,...,f_d',\zeta,\eta)$ and  $\hat{\mathcal{F}}^{(l)} := (\hat{f}_l^{(l)},...,\hat{f}_d'^{(l)},\hat{\zeta}^{(l)},\hat{\eta}^{(l)})$ for $l=1,...,d$.
\end{definition}

We stress that for $l=1,...,d$, conditioned on $\mathcal{F}^{(1)},...,\mathcal{F}^{(l)}$, the function $\hat{f}_j^{(l)}$ is still uniformly distributed in $P(S_j^{(l)},S_{j+1}^{(l)})$ for $j=l,...,d-1$. This is because the number of the ways to choose $\hat{\eta}^{(l)}$ is the same for any condition $\mathcal{F}^{(1)},...,\mathcal{F}^{(l)}$ and $\hat{f}_l^{(l)},...,\hat{f}_{d-1}^{(l)}$ in $\mathbf{BSHUF}(d,f)$. In this paper, similarly to \cite{ccl20}, we say that a quantum state $\rho$ or a classical string $s$ is uncorrelated to $\mathcal{F}^{(l)}$ if the process which outputs $\rho$ or $s$ will not change the output distribution even if we replace $\mathcal{F}^{(l)}$ by any other mapping. 

Following the concept of the hidden set $\mathbf{S}$, let us define the shadow. Definition \ref{shadow} is similar to Definition 5.3 in \cite{ccl20}, but, instead of $\perp$ (a symbol represents a constant with no information), we consider the value of the domain and a function $\xi$ which can be recognized as a boolean flag of the shadow. Note that the reason why we also consider the function $\xi$ is to make the quantum oracle unitary.

\begin{definition}[Shadow function]\label{shadow}
    Let $\mathcal{F}_b := (f_0,...,f_{d-1},f'_d,\zeta,\eta)$ be a $(d,f)$-bijective shuffling. Fix the hidden sets $\mathbf{S} = (\overline{S}^{(0)},...,\overline{S}^{(d)})$. Let $c$ be a constant bit independent of $\mathcal{F}_b$ and $\mathbf{S}$. For $j=l,...,d-1$, let $g_j$ be the function such that if $x \in S_j^{(l)}$, $g_j(x) = x$; otherwise, $g_j(x) = f_j(x)$. Let $g_d$ be the function such that if $x \in S_d^{(l)}$, $g_d(x) = x$; otherwise, $g_d(x) = f_d'(x)$. Let $\zeta_g$ be the function such that if $x \in S_d^{(l)}$, $\zeta_g(x) = c$; otherwise, $\zeta_g(x) = \zeta(x)$. Let $\eta_g$ be the function such that if $x \in S_d^{(l)}$, $\eta_g(x) = c$; otherwise, $\eta_g(x) = \eta(x)$. Finally, for $j=l,...,d$, let $\xi_j$ be the function such that if $x \in S_j^{(l)}$, $\xi_j=c$; otherwise, $\xi_j(x)=c \oplus 1$. The shadow $\mathcal{G}$ of $\mathcal{F}_b$ in $\overline{S}^{(l)} = (S_l^{(l)},...,S_d^{(l)})$ is defined as $\mathcal{G} := (f_0,...,f_{l-1},g_l,...,g_d,\xi_l,...,\xi_d,\zeta_g,\eta_g)$.
\end{definition}

We stress that the shadow $\mathcal{G}$ does not contain any information of $\mathcal{F}_b$ in $\overline{S}^{(l)}$, which is $\hat{\mathcal{F}}^{(l)}$. Now we define the quantum oracle $\mathcal{G}$ of $\mathcal{F}_b$ in $\overline{S}^{(l)}$ as follows:

\begin{eqnarray*}\nonumber
    \mathcal{G}\ket{\phi} &:=& \sum_{\mathbf{X}_0,...,\mathbf{X}_d} c(\mathbf{X}_0,...,\mathbf{X}_d) \biggl(\ket{0,\mathbf{X}_0}\ket{f_0(\mathbf{X}_0)} \otimes \bigotimes_{i=1}^{l-1} \ket{i, f_i(\mathbf{X}_i)}\ket{0} \otimes \\ && \mspace{70mu} \bigotimes_{i=l}^{d-1} \ket{i, g_i(\mathbf{X}_i)}\ket{\xi_i(\mathbf{X}_i)} \otimes \ket{d, g_d(\mathbf{X}_d)}\ket{\zeta_g(\mathbf{X}_d)}\ket{\eta_g(\mathbf{X}_d)} \ket{\xi_d(\mathbf{X}_d)} \biggl)_{\mathbf{R}_Q,\mathbf{R}_N} \\ && \mspace{400mu} \otimes \ket{w(\mathbf{X}_0,...,\mathbf{X}_d)}_{\mathbf{R}_W}.
\end{eqnarray*}

We will introduce the ``semi-classical''  oracle \cite{ahu18} in our setting.

\begin{definition}[$U^{\mathcal{F}_b \setminus \overline{S}^{(l)}}$, Definition 5.5 in \cite{ccl20}]
     Let $\mathcal{F}_b$ be a $(d,f)$-bijective shuffling. Let $\mathbf{S} = (\overline{S}^{(0)},...,\overline{S}^{(d)})$ be the hidden sets. Let $U$ be a unitary operator acts on qubits of the register $\mathbf{R}$. For $l=1,...,d$, let $U^{\mathcal{F}_b \setminus \overline{S}^{(l)}} := \mathcal{F}_b U_{\overline{S}^{(l)}}U$ be a unitary operator acts on qubits of registers $(\mathbf{R},\mathbf{I})$ where $\mathbf{I}$ is a one-qubit register and $U_{\overline{S}^{(l)}}$ is defined as follows:
    \begin{equation}\nonumber
        U_{\overline{S}^{(l)}} \ket{(l,\mathbf{X}_l),...,(d,\mathbf{X}_d)}_\mathbf{R} \ket{b}_\mathbf{I} := 
        \left\{
            \begin{alignedat}{2}   
                & \ket{(l,\mathbf{X}_l),...,(d,\mathbf{X}_d)}_\mathbf{R} \ket{b}_\mathbf{I} \text{ if every } \mathbf{X}_i \cap S_i^{(l)} = \emptyset, \\   
                & \ket{(l,\mathbf{X}_l),...,(d,\mathbf{X}_d)}_\mathbf{R} \ket{b \oplus 1}_\mathbf{I} \text{ otherwise}.
            \end{alignedat} 
        \right.
    \end{equation}
\end{definition}

\begin{definition}[$\mathrm{Pr}(find \, \, \overline{S}^{(k)}:U^{\mathcal{F}_b \setminus \overline{S}^{(k)}}, \rho)$, Definition 5.6 in \cite{ccl20}]\label{pr}
    Let $k,d \in \mathbb{N}$ satisfying $k \leq d$. Let $\rho$ be any input quantum state and $U$ be any unitary operator acts on $\rho$. We define
    \begin{equation}\nonumber
        \mathrm{Pr}[\text{find }\overline{S}^{(k)}:U^{\mathcal{F}_b \setminus \overline{S}^{(k)}}, \rho] := \mathbb{E} \left[ \mathrm{tr} \left( (I_\mathbf{R} \otimes (I - \ket{0} \bra{0})_\mathbf{I}) \circ U^{\mathcal{F}_b \setminus \overline{S}^{(k)}} \circ \rho \otimes \ket{0} \bra{0}_\mathbf{I} \right) \right],
    \end{equation}
    where $\mathbb{E}$ is the expectation value over the random $\mathcal{F}_b$ and $\overline{S}^{(k)}$.
\end{definition}

When $\rho$ is a pure state, say $\ket{\psi}$, we have
\[
    U^{\mathcal{F}_b \setminus \overline{S}^{(l)}} \ket{\psi}_\mathbf{R} \ket{0}_\mathbf{I} := \ket{\phi_0}_\mathbf{R} \ket{0}_\mathbf{I} + \ket{\phi_1}_\mathbf{R} \ket{1}_\mathbf{I}
\]
and $\mathrm{Pr}[find \, \, \overline{S}^{(k)}:U^{\mathcal{F}_b \setminus \overline{S}^{(k)}}, \ket{\psi}] = \mathbb{E}[||\ket{\phi_1}_\mathbf{R}||^2]$. Note that, due to the definition of $\mathcal{F}_b$, $\ket{\phi_0}$ and $\ket{\phi_1}$ are orthogonal because $\ket{\phi_0}$ involves no query to $\overline{S}^{(k)}$ but $\ket{\phi_1}$ does. Thus,
\[
    \mathcal{F}_bU\ket{\psi}=\ket{\phi_0}+\ket{\phi_1}.
\]
Similarly, let us consider
\[
    \mathcal{G}U_{\overline{S}^{(k)}}U \ket{\psi}_\mathbf{R} \ket{0}_\mathbf{I} := \ket{\phi_0}_\mathbf{R} \ket{0}_\mathbf{I} + \ket{\phi_1^{\perp}}_\mathbf{R} \ket{1}_\mathbf{I}.
\]
Since $\ket{\phi_0}$ and $\ket{\phi_1^{\perp}}$ are orthogonal from the definition of $\mathcal{G}$,
\[
    \mathcal{G}U\ket{\psi}=\ket{\phi_0}+\ket{\phi_1^{\perp}}.
\]
Note that $\ket{\phi_1}$ and $\ket{\phi_1^{\perp}}$ are orthogonal from the definition of $\mathcal{F}_b$ and $\mathcal{G}$. Then, by the concavity of the mixed state, we can prove the following lemma in a very similar way to Lemma 5.7 in \cite{ccl20}.

\begin{lemma}[Oneway-to-hiding(O2H) lemma for the bijective shuffling oracle]\label{O2H}
    Let $k,d \in \mathbb{N}$ satisfying $k \leq d$. Let $\mathcal{F}_b$ be a $(d,f)$-bijective shuffling. Let $\mathbf{S} = (\overline{S}^{(0)},...,\overline{S}^{(d)})$ be the hidden sets. Let $\mathcal{G}$ be the shadow of $\mathcal{F}_b$ in $\overline{S}^{(k)}$. Then, for any quantum algorithm $\mathcal{A}$, any unitary operator $U$, any initial quantum state $\rho$ and any classical string $t$,
    \begin{eqnarray*}
        |\mathrm{Pr}[\mathcal{A} (\mathcal{F}_bU(\rho)) = t] - \mathrm{Pr}[\mathcal{A} (\mathcal{G}U(\rho)) = t]| &\leq& B(\mathcal{F}_bU(\rho),\mathcal{G}U(\rho)) \\
        &\leq& \sqrt{2\mathrm{Pr}[\text{find }\overline{S}^{(k)}:U^{\mathcal{F}_b \setminus \overline{S}^{(k)}}, \rho]}.
    \end{eqnarray*}
\end{lemma}

By the union bound, it is also shown that the finding probability of $\overline{S}^{(k)}$ is bounded. The following lemma and its proof are very similar to Lemma 5.8 in \cite{ccl20} since
\[
    \left(\ket{0,\mathbf{X}_0}\ket{f_0(\mathbf{X}_0)} \otimes \bigotimes_{i=1}^{d-1} \ket{i, f_i(\mathbf{X}_i)}\ket{0} \otimes \ket{d, f_d'(\mathbf{X}_d)}\ket{\zeta(\mathbf{X}_d)}\ket{\eta(\mathbf{X}_d)}\ket{0} \right)
\]
are orthogonal for different sets of queries $\mathbf{X}_0,...,\mathbf{X}_d$.
\begin{lemma}\label{bfp}
    Suppose $\mathrm{Pr}[x \in S_i^{(k)}|x \in S_i^{(k-1)}] \leq p$ for $i=k,...,d$. Then for any unitary operator $U$ and initial quantum state $\rho$, which are promised to uncorrelated to $\hat{\mathcal{F}}^{(k-1)}$ and $\overline{S}^{(k)}$,
    \[
        \mathrm{Pr}[\text{find} \, \, \overline{S}^{(k)}:U^{\mathcal{F}_b \setminus \overline{S}^{(k)}}, \rho] \leq p \cdot q,
    \]
    where q is the number of parallel queries $\mathcal{F}_b U$ performs, which is $\sum_{i=0}^d |\mathbf{X}_i|$.
\end{lemma}

%% file: lowerbounds.tex
\section{Proof of Lower Bounds}\label{5}
In this section, we prove the lower bounds of $d$-$\mathsf{BSSP}$ for $(\mathsf{BQNC}_d)^\mathcal{O}$, $(\mathsf{BQNC}_d^{\mathsf{BPP}})^\mathcal{O}$ and  $(\mathsf{BPP}^{\mathsf{BQNC}_d})^\mathcal{O}$. The proofs are almost the same as the proof of Theorem 6.1, 7.1 and 8.1 in \cite{ccl20} respectively except applying Lemma \ref{O2H} and Lemma \ref{bfp} instead of Lemma 5.7 and Lemma 5.8 in \cite{ccl20}. 

\subsection{Lower bounds for $(\mathsf{BQNC}_d)^\mathcal{O}$}

First we establish $d$-$\mathsf{BSSP}$ is intractable for any $\mathsf{QNC}_d$ circuit. By applying the Oneway-to-Hiding lemma inductively, it can be shown that $U_{d+1} {\mathcal{F}_b} U_d \cdot \cdot \cdot {\mathcal{F}_b} U_1$ is indistinguishable from $U_{d+1} \mathcal{G} U_d \cdot \cdot \cdot \mathcal{G} U_1$. The same argument as in \cite{ccl20} gives the following result.

\begin{theorem}\label{lb_qncd}
    Let $n,d \in \mathbb{N}$. Let $\mathcal{A}$ be any $d$-depth quantum circuit and $\rho$ be any initial state. Let $f$ be a random Simon's function from $\mathbb{Z}_2^n$ to $\mathbb{Z}_2^n$ with period $s$ and $\mathcal{F}_b$ be the $(d,f)$-bijective shuffling sampled from $\mathcal{O}_{\mathrm{unif}}^{f,d}$. Then
    \[
        \mathrm{Pr}[\mathcal{A}^{\mathcal{F}_b}(\rho)=s] \leq d \cdot \sqrt{\frac{\mathsf{poly(n)}}{2^n}} + \frac{1}{2^n}.
    \]
\end{theorem}

\begin{corollary}
    The search $d$-$\mathsf{BSSP}$ can be solved with at most negligible probability by any $(\mathsf{QNC}_d)^\mathcal{O}$ and the decision $d$-$\mathsf{BSSP}$ is not in $(\mathsf{BQNC}_d)^\mathcal{O}$
\end{corollary}

\subsection{Lower bounds for $(\mathsf{BQNC}_d^{\mathsf{BPP}})^\mathcal{O}$}
Next, we establish the search $d$-$\mathsf{BSSP}$ is intractable by any $d$-QC scheme. Since each classical algorithms interspersed quantum layers can find polynomial-size number of points of each domains, a new procedure to define the hidden sets is needed. The same argument as in \cite{ccl20} gives the following result.

\begin{theorem}\label{th:lb_for_dqc}
    Let $n,d \in \mathbb{N}$. Let $\mathcal{A}$ be any $d$-QC scheme and $\rho$ be any initial state. Let $f$ be a random Simon's function from $\mathbb{Z}_2^n$ to $\mathbb{Z}_2^n$ with period $s$ and $\mathcal{F}_b$ be the $(d,f)$-bijective shuffling sampled from $\mathcal{O}_{\mathrm{unif}}^{f,d}$. Then
    \[
        \mathrm{Pr}[\mathcal{A}^{\mathcal{F}_b}(\rho)=s] \leq d \cdot \sqrt{\frac{\mathsf{poly(n)}}{2^n}}.
    \]
\end{theorem}

\begin{corollary}
    The search $d$-$\mathsf{BSSP}$ can be solved with at most negligible probability by any $d$-QC scheme and the decision $d$-$\mathsf{BSSP}$ is not in $(\mathsf{BQNC}_d^{\mathsf{BPP}})^\mathcal{O}$
\end{corollary}

\subsection{Lower bounds for $(\mathsf{BPP}^{\mathsf{BQNC}_d})^\mathcal{O}$}

Finally, we establish the search $d$-$\mathsf{BSSP}$ is hard for any $d$-CQ scheme. The main difficulty lies in that conditioned on measurement results of quantum circuits, the distribution of the permutations might be not uniform enough to prove the hardness in a similar way. Chia et al. resolved the problem by approximating it by a convex combination of ``almost'' uniform shufflings. The same argument as in \cite{ccl20} gives the following result.

\begin{theorem}\label{th:lb_for_dcq}
    Let $n,d \in \mathbb{N}$. Let $\mathcal{A}$ be any $d$-CQ scheme. Let $f$ be a random Simon's function from $\mathbb{Z}_2^n$ to $\mathbb{Z}_2^n$ with period $s$ and $\mathcal{F}_b$ be the $(d,f)$-bijective shuffling sampled from $\mathcal{O}_{\mathrm{unif}}^{f,d}$. Then
    \[
        \mathrm{Pr}[\mathcal{A}^{\mathcal{F}_b}()=s] \leq d \cdot \sqrt{\frac{\mathsf{poly(n)}}{2^n}}.
    \]
\end{theorem}

\begin{corollary}
    The search $d$-$\mathsf{BSSP}$ can be solved with at most negligible probability by any $d$-CQ scheme and the decision $d$-$\mathsf{BSSP}$ is not in $(\mathsf{BPP}^{\mathsf{BQNC}_d})^\mathcal{O}$.
\end{corollary}